\documentclass[letterpaper,11pt]{article}
\usepackage{fullpage}
\usepackage[english]{babel}
\usepackage{graphicx}
\usepackage{color}
\usepackage{amsmath,amsfonts,amssymb,amsthm,graphicx,graphics,epsf}
\usepackage{xspace}
\usepackage{hyperref}
\usepackage{wrapfig}
\usepackage{refcount}
\usepackage{setspace, float}

\graphicspath{{img/}} 

\newtheorem{defin}{Definition}
  
\newtheorem{theo}[defin]{Theorem}
  \newenvironment{theorem}{\begin{theo} \sl}{\end{theo}}
\newtheorem{lem}[defin]{Lemma}
  \newenvironment{lemma}{\begin{lem} \sl}{\end{lem}}
\newtheorem{coro}[defin]{Corollary}
  
\newtheorem{obs}[defin]{Observation}

\newcommand{\pathto}{\ensuremath{\to}}

\newcommand{\isink}{$i$-sink\xspace}
\newcommand{\zerosink}{$0$-sink\xspace}

\newcommand{\zerosinks}{$0$-sinks\xspace}

\newcommand{\etal}{\emph{et al.}\xspace} 

\newcommand{\figref}[1]{Figure~\ref{#1}}

\newcommand{\icone}{$i$-cone\xspace}
\newcommand{\zerocone}{$0$-cone\xspace}
\newcommand{\zerocones}{$0$-cones\xspace}

\newcommand{\YaoThree}{\ensuremath{\mathrm{Y_3}}}

\title{Theta-3 is connected\footnote{A preliminary version of this paper appeared in the proceedings of the 25th Canadian Conference on Computational Geometry (CCCG 2013) \cite{aichholzer2013theta3}}}
\date{}

\author{
Oswin Aichholzer\thanks{Institute for Software Technology, Graz University of Technology.}
\and Sang Won Bae \thanks{Department of Computer Science, Kyonggi University.}
\and Luis Barba \thanks{School of Computer Science, Carleton University.}\ \ 
\thanks{Boursier FRIA du FNRS, D\'epartement d'Informatique, Universit\'e Libre de Bruxelles.}
\and Prosenjit Bose \footnotemark[3]
\and Matias Korman \thanks{National Institute of Informatics, Tokyo, Japan.}\ \ \thanks{JST, ERATO, Kawarabayashi Large Graph Project.}
\and Andr\'e van Renssen \footnotemark[3]
\and Perouz Taslakian \thanks{College of Science and Engineering, American University of Armenia.}
\and Sander Verdonschot \footnotemark[3]
}

\begin{document}
\thispagestyle{empty}
\maketitle

\begin{abstract}
In this paper, we show that the $\theta$-graph with three cones is connected.
We also provide an alternative proof of the connectivity of the Yao graph with three cones.
\end{abstract}

\section{Introduction}
Introduced independently by Clarkson~\cite{clarkson1987approximation} in 1987 and Keil~\cite{keil1988approximating} in 1988, the $\theta$-graph of a set $P$ of points in the plane is constructed as follows. We consider each point $p\in P$ and partition the plane into $m \geq 2$ cones (regions in the plane between two rays originating from the same point) with apex $p$, each defined by two rays at consecutive multiples of $2\pi/m$ radians from the negative $y$-axis; see \figref{fig:Construction of Theta 3} for an illustration. We label the cones $C_0$ through $C_{m-1}$, in clockwise order around $p$, starting from the cone whose angular bisector aligns with
the positive $y$-axis from $p$ if $m$ is odd, or having this axis as its left boundary if $m$ is even. If the apex is not clear from the context, we use $C_i^p$ to denote the cone $C_i$ with apex $p$. We sometimes refer to $C_i^p$ as the \emph{$i$-cone} of $p$. 
To build the $\theta$-graph, we consider each point $p$ and connect it by an edge with the \emph{closest} point in each of its cones. 
However, instead of using the Euclidean distance, we measure distance by orthogonally projecting each point onto the angle-bisector of that cone. 
The \emph{closest} point to $p$ in its $i$-cone is then the point in $C_i^p$ whose projection has the smallest Euclidean distance to $p$.

We use this definition of distance in the remainder of the paper, except for Section~\ref{sec:Yao}, which deals with Yao graphs. 
For simplicity, we assume that no two points of $P$ lie on a line parallel to the boundary of a cone or perpendicular to the angular bisector of a cone, guaranteeing that each point connects to at most one point in each cone. 
We call the $\theta$-graph with $m$ cones the $\theta_m$-graph.

\begin{figure}[t]
\centering
\includegraphics{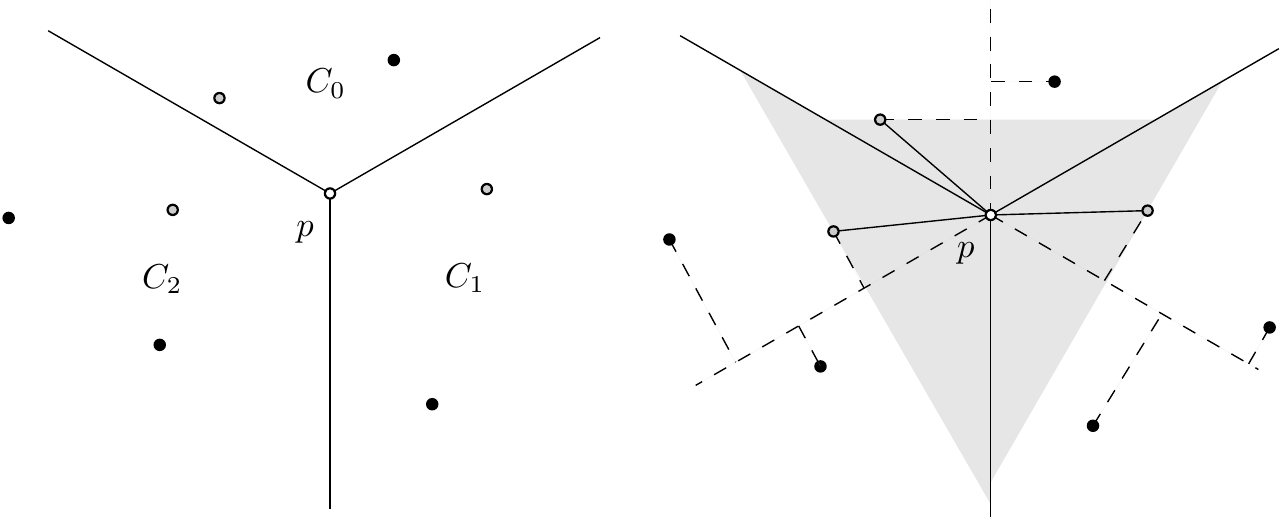}
\caption{\small Left: A point $p$ and its three cones in the $\theta_3$-graph. Right: Point $p$ adds an edge to the closest point in each of its cones, where distance is measured by projecting points onto the bisector of the cone.}
\label{fig:Construction of Theta 3}
\end{figure}

For $\theta$-graphs with an even number of cones, proving connectedness is easy. As the first $m/2$ cones cover exactly the right half-plane, each point will have an edge to a point to its right, if such a point exists. Thus, we can find a path from any point to the rightmost point and, by concatenating these, a path between every pair of points. Unfortunately, if $m$ is odd this property does not hold, as no set of cones covers \emph{exactly} the right half-plane. Therefore, a point is not guaranteed to have an edge to a point to its right, even if such a point exists.

The fact that $\theta$-graphs with more than 6 cones are connected has been known for a long time. In fact, they even guarantee the existence of a \emph{short} path between every pair of points. The length of this path is bounded by a constant times the straight-line Euclidean distance between the two points~\cite{bose2012optimal,bose2013spanning,clarkson1987approximation,keil1988approximating,ruppert1991approximating}. Graphs that have this property are called \emph{geometric spanners}. For more information on geometric spanners, see the book by Narasimhan and Smid~\cite{NS06}.

For a long time, very little was known about $\theta$-graphs with fewer than 7 cones. Bonichon~\etal~\cite{bonichon2010connections} broke ground in this area in 2010, by showing that the $\theta_6$-graph is a geometric spanner. Subsequently, both the $\theta_4$- and $\theta_5$-graphs have been shown to be geometric spanners~\cite{barba2013stretch, bose2013theta5}. El~Molla~\cite{el2009yao} already showed that the $\theta_2$- and $\theta_3$-graphs are not geometric spanners. 
It is straightforward to verify that the $\theta_2$-graph is connected which leaves the $\theta_3$-graph as the only $\theta$-graph for which connectedness has not been proven. 
In this paper, we settle this question by showing that the $\theta_3$-graph is always connected.

The question of connectedness about the $\theta_3$-graph is interesting because the $\theta_3$-graph has some unique properties that cause standard proof techniques for $\theta$-graphs to fail. As such, we hope that the techniques we develop here will lead to more insight into the structure of other $\theta$-graphs. As an example, most proofs for a larger number of cones show that the $\theta$-routing algorithm (always follow the edge to the closest vertex in the cone that contains the destination) returns a short path between any two points. But in the $\theta_3$-graph, $\theta$-routing is not guaranteed to ever reach the destination. The smallest point set that exhibits this behavior has three points, such that for each point, both other points lie in the same cone; see \figref{fig:notstrong}. In fact, this example shows not only that this exact routing strategy fails; it shows that if we consider the edges to be directed (from the point that added them, to the closest point in its cone), the graph is not strongly connected. Therefore, our proof requires more global methods than previous proofs on $\theta$-graphs.

\begin{figure}[t]
\centering
\includegraphics{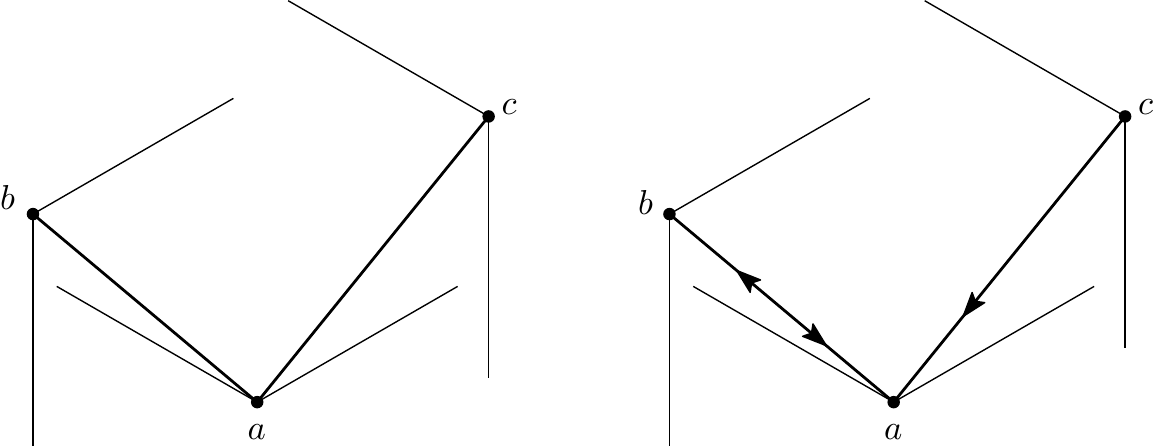}
\caption{\small Left: A point set for which $\theta$-routing does not find a path from $a$ to $c$, as it keeps cycling between $a$ and $b$.
Right: The directed version of the graph is not strongly connected, as there is no path from either $a$ or $b$ to $c$.}
\label{fig:notstrong}
\end{figure}

Most proofs for a larger number of cones use induction on the distance between points or on the size of the empty triangle between a point and its closest point. In the $\theta_3$-graph however, both of these measures can increase when we follow an edge. Thus, applying induction on these distances seems a difficult task. An induction on the number of points similarly fails, as inserting a new point may remove edges that were present before, and it is not obvious that the endpoints of those edges are still connected in the new graph.

The $\theta_3$-graph is strongly related to the \YaoThree-graph, where each point also connects to the closest point in each cone, but the distance measure is the standard Euclidean distance. This graph was shown to be connected by Damian and Kumbhar~\cite{damian2011undirected}. Their proof uses induction on a rhomboid distance-measure that was tailored specifically for the \YaoThree-graph. Since the `closest' point for the $\theta_3$-graph can be much further away than in the \YaoThree-graph, this method of induction does not translate to the $\theta_3$-graph, either. Conversely, we show that our proof extends to the \YaoThree-graph, providing an alternative proof for its connectivity.

\section{Properties of the $\boldsymbol{\theta_3}$-graph}

For $i\in \{0, 1, 2\}$, the edge connecting a point with its closest point in cone $C_i$ is called an \emph{$i$-edge}. Note that an edge can have one or two roles depending on the position of its endpoints. 
An example is depicted in \figref{fig:notstrong}, where edge $ab$ is both the $0$-edge of $a$ and the $1$-edge of $b$.

\begin{lemma}\label{lemma:Planar}
For all $i\in \{0, 1, 2\}$, no two $i$-edges of the $\theta_3$-graph can cross.
\end{lemma}
\begin{proof}
We consider only 0-edges of $P$; the proof is analogous for 1- and 2-edges. 
For a contradiction, assume that there are two 0-edges that cross at a point $s$. 
Call these edges $u_1 v_1$ and $u_2 v_2$, such that $v_1$ is in the \zerocone of $u_1$ and $v_2$ in the \zerocone of $u_2$. 
Assume without loss of generality that the $y$-coordinate of $v_1$ is smaller than that of $v_2$; see \figref{fig:Edges cannot cross} for an illustration.
Because $s$ lies on segments $u_1v_1$ and $u_2 v_2$, $s$ lies in the \zerocones of both $u_1$ and $u_2$.
Therefore, the \zerocone of $s$ is contained in the intersection of the \zerocones of $u_1$ and $u_2$. 
As $v_1$ lies in cone $C_0$ of $s$, point $v_1$ lies in cone $C_0$ of $u_2$ as well. 
Because we assumed that the $y$-coordinate of $v_1$ is less than that of $v_2$, we conclude that $v_1$ is closer to $u_2$ than $v_2$.
Thus, the edge $u_2 v_2$ is not a 0-edge, yielding a contradiction.
\end{proof}

\begin{figure}[h!]
\centering
\includegraphics{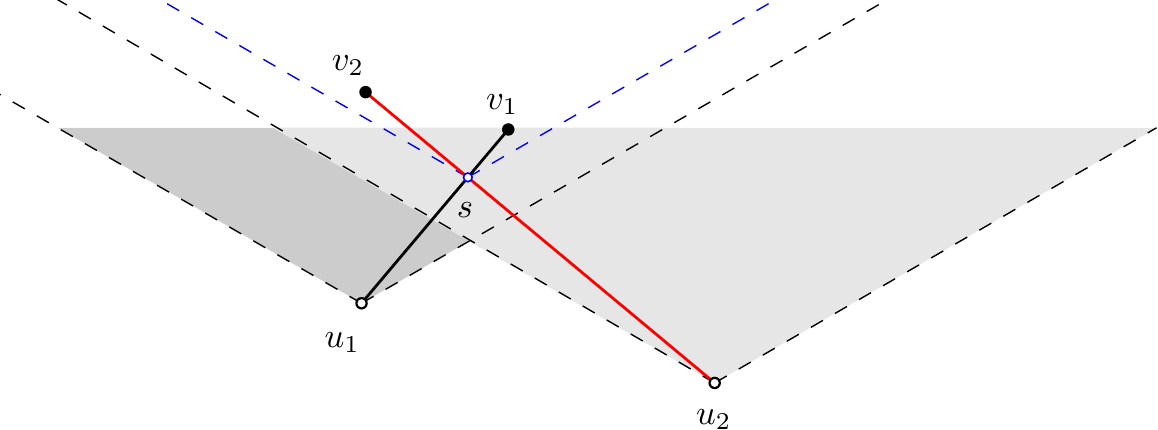}
\caption{\small Two 0-edges $u_1v_1$ and $u_2v_2$ such that $v_1\in C_0^{u_1}$ and $v_2\in v_1\in C_0^{u_2}$ cannot cross because the lowest point among $v_1$ and $v_2$ will be adjacent to both $u_1$ and $u_2$.}
\label{fig:Edges cannot cross}
\end{figure}

We say that a cone is \emph{empty} if it contains no point of $P$ in its interior.
A point having an empty \icone is called an \emph{\isink}.

Given a point $p$ of $P$, the \emph{$i$-path from $p$} is defined recursively as follows:
If the \icone of $p$ is empty, the $i$-path from $p$ consists of the single point $p$.
Otherwise, let $q$ be the closest point to $p$ in its \icone. 
The $i$-path from $p$ is defined as the union of edge $pq$ with the $i$-path from $q$.

\begin{lemma}\label{lemma:Forest}
Every $i$-path of the $\theta_3$-graph is well-defined and has an \isink at one of its endpoints.
\end{lemma}
\begin{proof}
We consider only 0-paths; the proof is analogous for the other paths.
A 0-path from a point $p$ is well defined because the closest point in the \zerocone of $p$ always lies above $p$.
Therefore, the $y$-coordinates of the points in the 0-path from $p$ form a monotonically increasing sequence.
As $P$ is a finite set, the recursion must end at a point having an empty \zerocone.
\end{proof}

\begin{figure}[hb]
\centering
\includegraphics{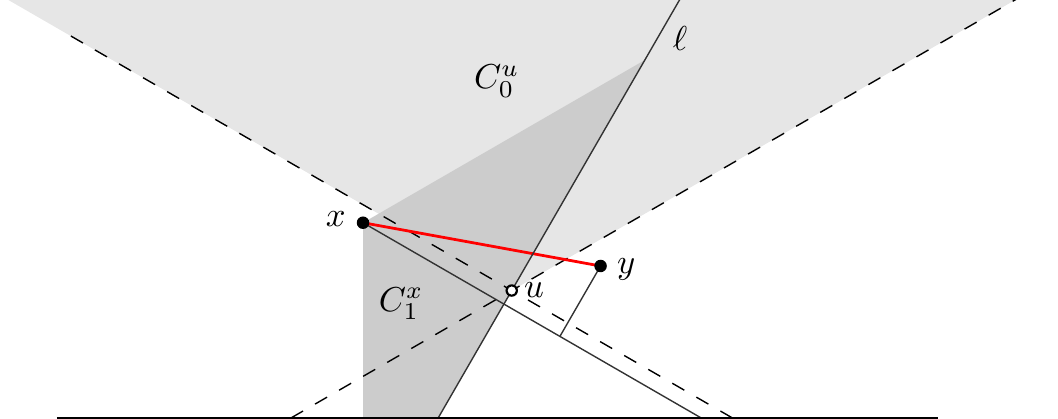}
\caption{\small Empty cones cannot be crossed by edges of the $\theta_3$-graph.}
\label{fig:Empty Cones Cannot be crossed}
\end{figure}

\begin{lemma}\label{lemma:Empty cones cannot be crossed}
If a cone of a point is empty, then no edge of the $\theta_3$-graph can cross this cone.
\end{lemma}
\begin{proof}
We consider only 0-cones for this proof; analogous arguments hold for the other cones.
Let $u$ be a point of $P$ with an empty 0-cone.
For a contradiction,
assume that there exists an edge $x y$ that crosses $C_0^u$. 
For this to happen, $x$ and $y$ have to lie in opposite sectors of the double wedge obtained by extending the boundary segments of $C_0^u$; see \figref{fig:Empty Cones Cannot be crossed}.
Assume without loss of generality that $x$ lies in the left wedge. Then $x$ lies in $C_2^u$ while $y$ lies in $C_1^u$.
In particular, this implies that both $u$ and $y$ lie in $C_1^x$.

Let $\ell$ be the line through $u$ perpendicular to the bisector of $C_1^x$. 
For the edge $x y$ to exist, the projection of $y$ on the bisector of $C_1^x$ must be closer to $x$ than the projection of $u$. In other words, $y$ must lie to the left of $\ell$.
However, all points lying to the left of $\ell$ are contained in $C_0^u \cup C_2^u$, yielding a contradiction as $y \in C_1^u$.
\end{proof}

\begin{figure}[t]
\centering
\includegraphics{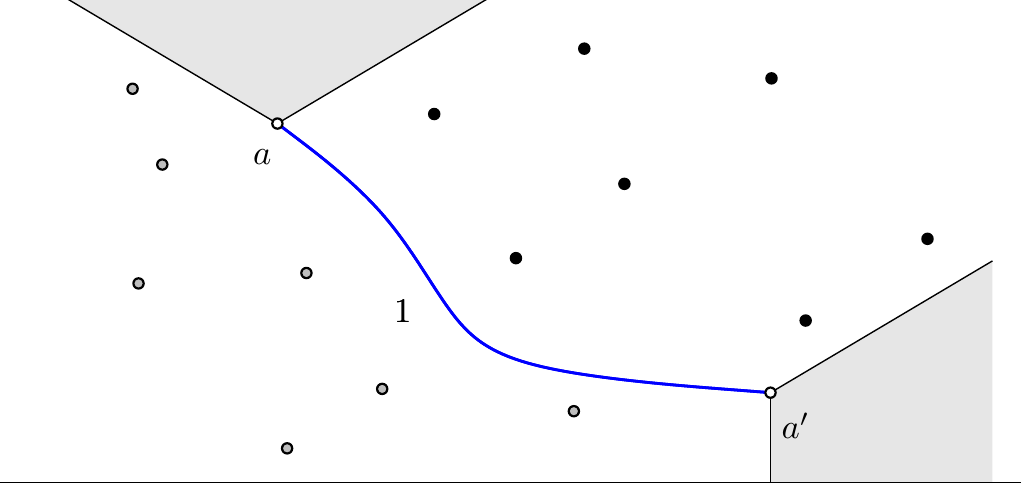}
\caption{\small A 1-barrier, defined by the 1-path joining $a$ with $a'$, splits the remaining points into two sets such that no two points in different sets can be joined by a 1-path.}
\label{fig:Barriers}
\vspace{-1em} 
\end{figure}

As a consequence of Lemmas~\ref{lemma:Planar} and \ref{lemma:Empty cones cannot be crossed}, two sinks connected by an $i$-path partition the remaining points into two sets such that no $i$-path can connect a point in one set to a point in the other set, as any such path would cross either the $i$-path between the sinks, or the empty cone of one of the sinks. 
Such a construction is called an \emph{$i$-barrier}; see \figref{fig:Barriers} for an illustration.

\section{Proving connectedness}
In this section we prove that the $\theta_3$-graph of any given point set is connected.
We start by proving that three given \zerosinks in a specific configuration are always connected. 
We then prove that  if the $\theta_3$-graph has at least two disjoint connected components, there exist three 0-sinks that are in this configuration and are not all in the same component, leading to a contradiction. 

Although the edges of the $\theta_3$-graph are not directed, by Lemma~\ref{lemma:Forest} we can think of an $i$-path as oriented towards the \isink it reaches.
An $i$-path from $a$ that ends at an $i$-sink $b$ is denoted by $a \pathto b$.
The following lemma is depicted in \figref{fig:Zero Sink configuration}.

\begin{figure}[b]
\centering
\includegraphics{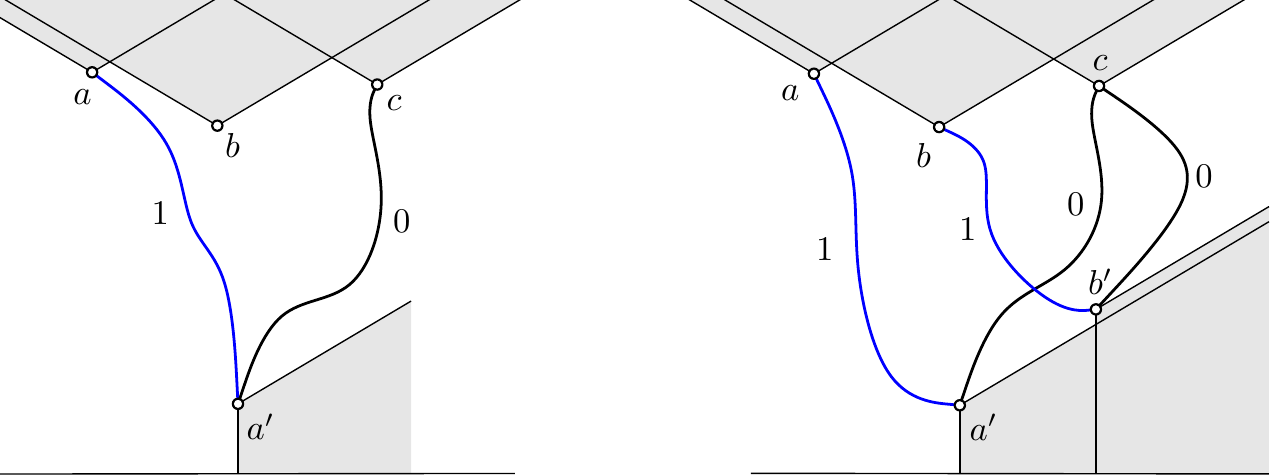}
\caption{\small Left: The configuration of points described in Lemma~\ref{lemma:0-sink configuration}.
Right: The configuration in the base case of the induction where no 0-sink lies to the right of $c$. 
}
\label{fig:Zero Sink configuration}
\end{figure}

\begin{lemma}\label{lemma:0-sink configuration}
Let $a$, $b$, and $c$ be three 0-sinks such that (i) $a$ lies to the left of $b$ and $b$ lies to the left of $c$, and (ii) the 1-path from $a$ ends at a 1-sink $a'$ whose 0-path ends at $c$ ($a'$ may be equal to $c$). Then, $a$, $b$, and $c$ belong to the same connected component.
\end{lemma}
\begin{proof}
Because there is a path from $a$ to $c$ via $a'$, $a$ and~$c$ must be in the same component. We show that $b$ belongs to this same connected component. 

The proof proceeds by induction on the number of 0-sinks to the right of $c$.
In the base case, there are no 0-sinks to the right of~$c$.
Consider the 1-sink $b'$ at the end of the 1-path from $b$; see \figref{fig:Zero Sink configuration} (right).
Because the 1-path $a \pathto a'$ forms a 1-barrier, $b'$ cannot lie to the left of $a'$. 

If $a' = c$, then $a'$ is both a 1-sink and a 0-sink.
This means that there can be no points to the right of $a'$.
Therefore $b'$ must also be equal to $a'$.
But then $b$ is in the same connected component as $a$ and we are done.
So assume that this is not the case, that is, $a'\neq c$ and $b'$ lies \mbox{to the right of $a'$}.

Then the 1-path $b\pathto b'$ also has to cross the 0-path $a' \pathto c$, as otherwise $a' \pathto c$ crosses the empty cone of $b'$, which is impossible by Lemma~\ref{lemma:Empty cones cannot be crossed}, or $b'$ lies on $a' \pathto c$ and we are done.
Moreover, because $a' \pathto c$ forms a 0-barrier, the 0-path from $b'$ cannot end to the left of $c$.
However, since there are no 0-sinks to the right of $c$, the 0-path from $b'$ must end at $c$.
Thus, there is a path connecting $b$ and $c$, which proves the lemma in the base case.

\begin{figure}[ht]
\centering
\includegraphics{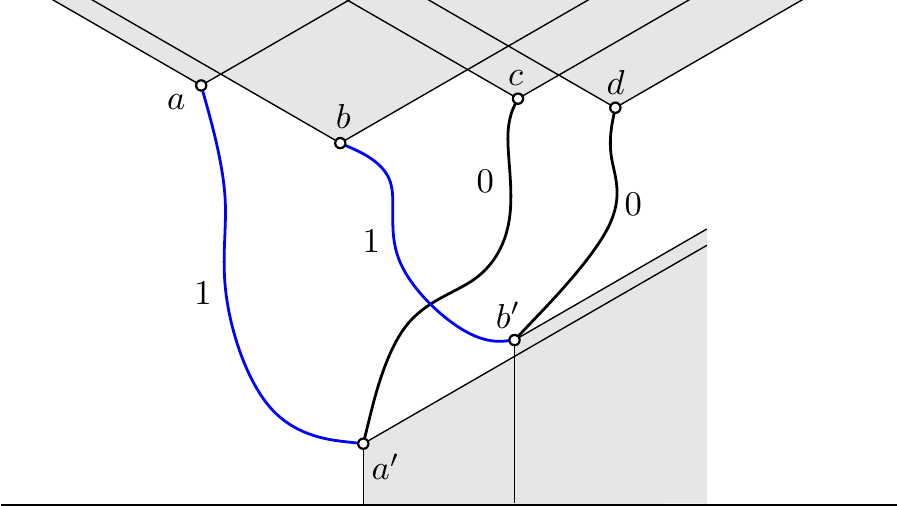}
\caption{\small The configuration of the inductive step where the induction hypothesis can be applied on 0-sinks $b$, $c$ and $d$.}
\label{fig:ZeroSink Inductive step}
\end{figure}

For the inductive step, let $k$ be the number of 0-sinks to the right of $c$ and assume that the lemma holds for any triple of 0-sinks with fewer than $k$ 0-sinks to their right. By the same argument as in the base case, we have a 1-path from $b$ to a 1-sink $b'$ that lies to the right of $a'$. Now consider the 0-sink $d$ at the end of the 0-path from $b'$; see \figref{fig:ZeroSink Inductive step}. Note that $b'$ and $d$ could be the same vertex.

Since the 0-path $a' \pathto c$ forms a 0-barrier, $d$ cannot lie to the left of $c$. 
If $d$ and $c$ are the same point, we have a path connecting $b$ and $c$ as in the base case, so assume that this is not the case. Thus $d$ lies to the right of $c$. Now $b$, $c$, and $d$ form a triple of 0-sinks that satisfy criteria (i) and (ii). And since $d$ is a 0-sink to the right of $c$, there are fewer than $k$ 0-sinks to the right of $d$. Thus, by induction, we have that $b$ is in the same connected component as $c$, which proves the lemma.
\end{proof}

\begin{theorem}\label{thm:Theta 3 is connected}
The $\theta_3$-graph is connected.
\end{theorem}
\begin{proof}
Assume for a contradiction that there exists a point set $P$ whose $\theta_3$-graph $G$ is not connected. 
From each point, we can follow its 0-path to a 0-sink. Therefore, $G$ must contain at least one 0-sink for each connected component. 
Let $a$ be the leftmost 0-sink, and let $A$ be the connected component of $G$ that contains $a$.
Now let $b$ be the leftmost 0-sink that does not belong to $A$.

We use Lemma~\ref{lemma:0-sink configuration} to show that, in fact, $b$ must belong to $A$ as well. 
Before we can do this, we need to define two barriers. 
The first barrier is formed by the 2-path from $b$, ending at a 2-sink $b'$. Because $a$ lies in $C_2^b$, point~$b$ does not have an empty 2-cone and hence, $b'$ differs from $b$. The second barrier is formed by the 0-path from $b'$, which ends at a 0-sink $c$; see \figref{fig:Connectedness}. Since $b$ is the leftmost 0-sink that does not belong to $A$, either $c$ and~$b$ are the same point, or $c$ lies to the right of $b$.

Now consider the 1-sink $a'$ at the end of the 1-path from $a$. 
This point has to lie to the right of both barriers $b\pathto b'$ and $b'\pathto c$, as otherwise these paths would cross the empty cone $C_1$ of $a'$, which is not allowed by Lemma~\ref{lemma:Empty cones cannot be crossed}. 
Because the path $a\pathto a'$ is a 1-path and the barriers in question consist of 0- and 2-edges, these crossings are possible. 
Now let $d$ be the 0-sink at the end of the 0-path from $a'$. 
Since this path cannot cross the 0-barrier $b'\pathto c$, $d$ cannot lie to the left of $c$. 

\begin{figure}[t]
\centering
\includegraphics{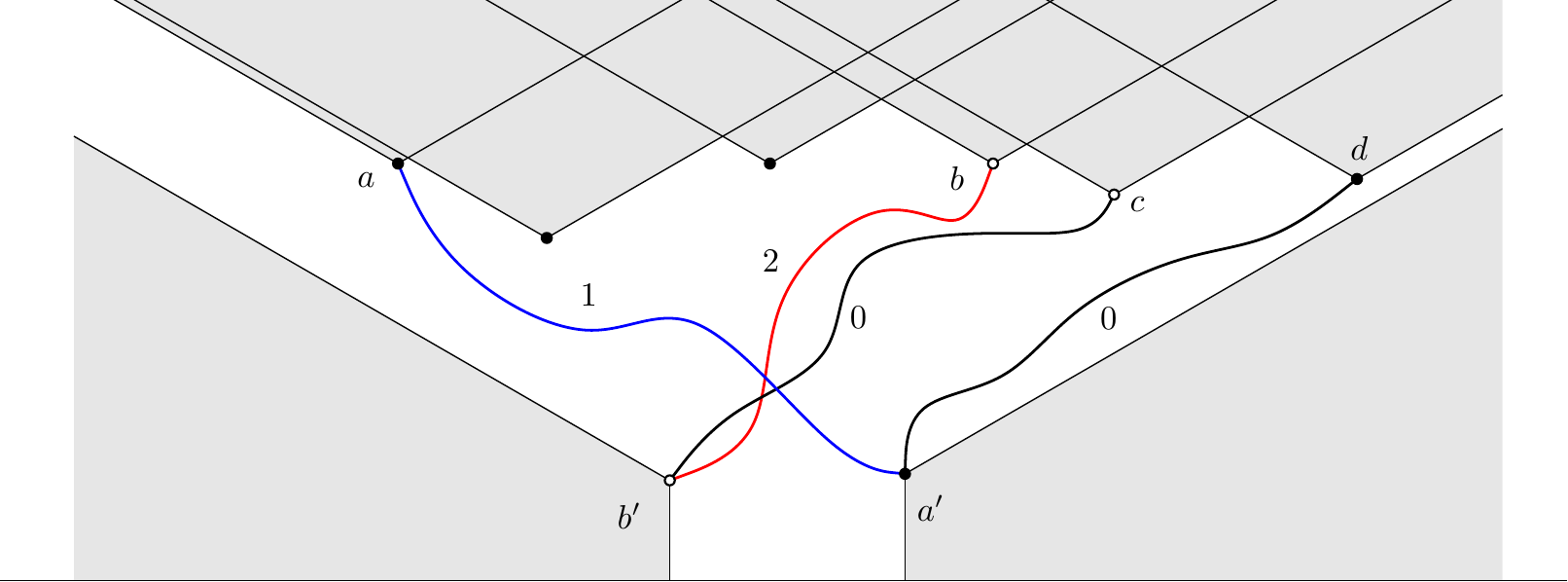}
\caption{\small Two \zerosinks $a$ and $b$ are assumed to lie in different components such that both $a$ and $b$ are the leftmost 0-sinks in their component.
The 1-path from $a$ ends at a 1-sink $a'$ whose 0-path ends at a \zerosink $d$ lying to the right of $b$. The 0-sinks $a, b$ and $d$ jointly satisfy the criteria of Lemma~\ref{lemma:0-sink configuration}.}
\label{fig:Connectedness}
\end{figure}

Because $d$ belongs to component $A$, if $c$ and $d$ are the same point, $c$ belongs to component $A$. Otherwise, if $c$ and $d$ are distinct points, then $a$, $b$, and $d$ jointly satisfy the criteria of Lemma~\ref{lemma:0-sink configuration}, which gives us that $b$ belongs to component $A$ as well---a contradiction since~$b$ is the leftmost \zerosink that does not belong to $A$.
This contradiction comes from our assumption that $G$ is not connected.
Therefore, the $\theta_3$-graph of any point set is connected.
\end{proof}

\section{The $\boldsymbol{\YaoThree}$-graph}\label{sec:Yao}
The construction of the \YaoThree-graph is very similar to that of the $\theta_3$-graph. The only difference is the way distance is measured: the $\theta$-graph uses the length of the projection onto the bisector, whereas the Yao graph uses the Euclidean distance. Therefore, in every cone a point is connected to its closest Euclidean neighbor. We denote by $|pq|$ the Euclidean distance between two points $p$ and $q$.

We show that, like the $\theta_3$-graph, the \YaoThree-graph is connected. To this end, we re-introduce the three basic lemmas we had for the $\theta_3$-graph and show that the same properties hold for the \YaoThree-graph. 
We first prove a geometric auxiliary lemma depicted in \figref{fig:Yao3Aux}.
  \begin{figure}[t]
    \centering
    \includegraphics{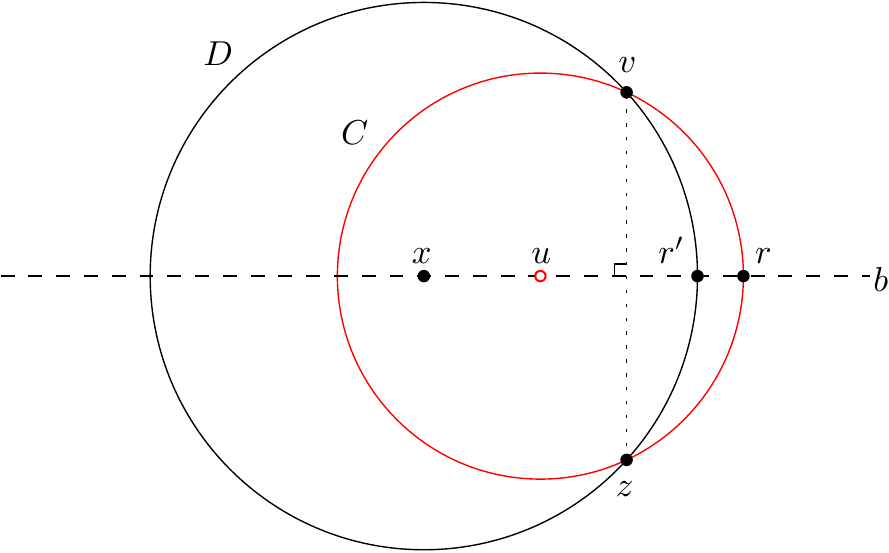}
    \caption{\small Point $x$ lies to the left of point $u$ and the arcs $v r'$ and $r' z$ are enclosed by circle $C$ centered at $u$, having radius $|uv|$.}
    \label{fig:Yao3Aux}
  \end{figure}

\pagebreak
\begin{lemma}
  \label{lemma:Yao3Aux}
Given a non-vertical line $b$ and a circle $C$ centered at a point $u$ on $b$,
 let $v$ and $z$ be two points on $C$ such that $b$ bisects the segment $vz$.
Let $x$ be a point on $b$ and let $D$ be the circle centered at $x$ with radius~$|x v|$. 
If $x$ lies to the left of $u$, then the right-side arc of $D$ between $v$ and $z$
 is enclosed by $C$; otherwise, the left-side arc of $D$ between $v$ and $z$
 is enclosed by $C$. 
\end{lemma}
\begin{proof}
Assume that $x$ lies to the left of $u$; the proof of the other case is analogous. 
Let $r$ and $r'$ be the respective right intersections of $C$ and $D$ with line $b$; see \figref{fig:Yao3Aux}. 
Hence, arcs $v r'$ and $r' z$ lie either entirely inside $C$ or entirely outside $C$. Therefore, it suffices to show that $r'$ lies inside $C$, i.e., $|u r'| \leq |u r|$. Since $x$ lies to the left of $u$, we can rewrite $|u r'|$ as $|x r'| - |x u|$. Since $|x r'| = |x v|$ and $|u r| = |u v|$, we thus need to show that $|x v| \leq |x u| + |u v|$. This follows from the triangle inequality. 
\end{proof}

The proof of the following lemma is similar to that of Lemma~\ref{lemma:Planar}.

\begin{lemma}
  \label{lemma:Yao3Planar}
For all $i\in \{0, 1, 2\}$, no two $i$-edges of the \YaoThree-graph can cross.
\end{lemma}
\begin{proof}
  We look at the 0-edges. The cases for the other edges are analogous. Let $u v$ be a 0-edge such that $v\in C_0^u$ and assume without loss of generality that $v$ lies to the right of $u$. We prove the lemma by contradiction, so assume that some 0-edge $x y$  crosses $u v$ and let $y \in C_0^x$. Note that for $x y$ to cross $u v$, $C_0^x$ must contain some part of $u v$. Hence $v$ lies in $C_0^x$. 

Let $k$ be the line through the right boundary of $C_0^u$ and let $l$ be the line through $u$, perpendicular to $k$.
We consider four cases, depending on the location of $x$ with respect to $u$; see \figref{fig:Yao3PlanarForestsA} (left): (a) $x \in C_0^u$ to the left of the line $u v$, (b) $x \in C_2^u$ above $k$, (c) $x \in C_2^u$ below $k$ or $x \in C_1^u$ below $l$, (d) $x \in C_1^u$ above $l$ or $x \in C_0^u$ to the right of the line $u v$. 

  \begin{figure}[t]
    \centering
    \includegraphics{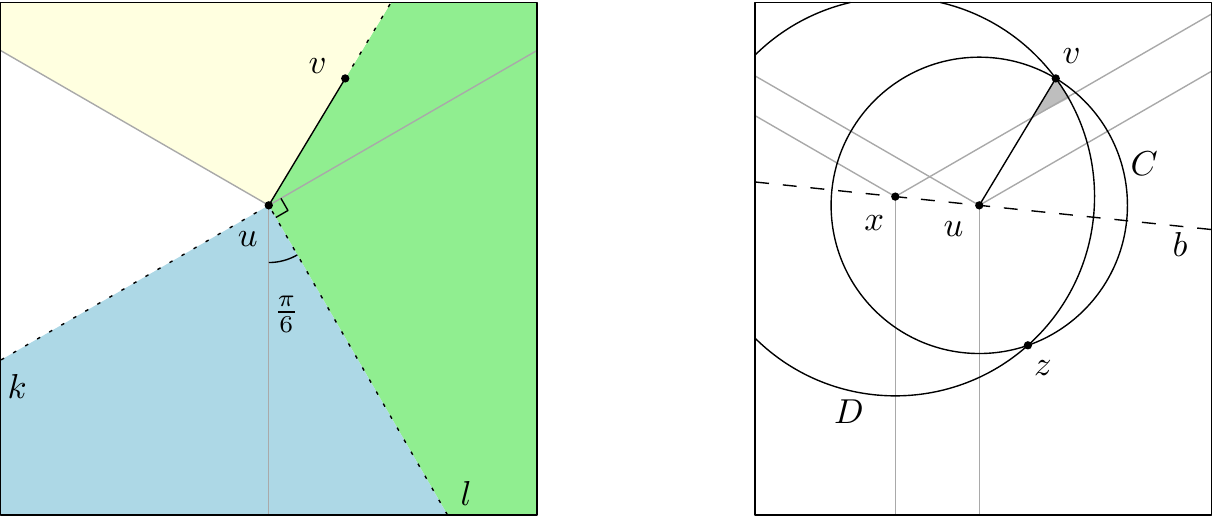}
    \caption{\small Left: The four cases. Right: The case when $x$ lies in $C_2^u$ and above $k$.}
    \label{fig:Yao3PlanarForestsA}
  \end{figure}

  \textbf{Case (a):} $x \in C_0^u$ to the left of the line $u v$. Since~$v$ lies inside $C_0^x$ and $v$ lies to the right of $u$, $x$ lies in the circle centered at $u$ having radius $|u v|$. Thus, $x$ lies closer to $u$ than $v$, contradicting the existence of edge~$u v$. 

  \textbf{Case (b):} $x \in C_2^u$ above $k$. We apply Lemma~\ref{lemma:Yao3Aux} as follows, see \figref{fig:Yao3PlanarForestsA} (right): Let $C$ be the circle centered at $u$ having radius $|u v|$. 
Let $b$ be the line through $u$ and $x$, and let $z$ be the reflection of $v$ in $b$.  
Note that this implies that $z$ lies outside $C_0^u$. Let $D$ be the circle centered at $x$ having radius $|x v|$. Since $x$ lies to the left of $u$, Lemma~\ref{lemma:Yao3Aux} gives us that the right arc $v z$ of circle $D$ is enclosed by circle $C$. Since the area in which $y$ must lie for $x y$ to cross $u v$ is bounded by the right boundary of $C_0^x$, edge $u v$, and the right arc $v z$ of circle $D$, it is enclosed by $C$. 
Therefore, any such point would lie in $C_0^u$ and be closer to $u$ than $v$, contradicting the existence of edge $uv$.

  \textbf{Case (c):} $x \in C_2^u$ below $k$ or $x \in C_1^u$ below $l$; see \figref{fig:Yao3PlanarForestsB} (left). Since~$u$ lies in $C_0^x$, $y$ needs to be closer to $x$ than $u$ for edge $x y$ to exist. Hence it must lie inside the circle $C$ centered at $x$ with radius $|x u|$. Look at the lower half-plane defined by the line tangent to $C$ at $u$ and note that $C$ is contained in this half-plane. However, the half-plane does not intersect $C_0^u$ to the right of $u$ and hence no point $y$ inside the half-plane can be used to form an edge $x y$ that crosses~$u v$. 

  \textbf{Case (d):} $x \in C_1^u$ above $l$ or $x \in C_0^u$ to the right of the line $u v$. We apply Lemma~\ref{lemma:Yao3Aux} as follows, see \figref{fig:Yao3PlanarForestsB} (right): Let $C$ be the circle centered at $u$ having radius~$|u v|$. Let $b$ be the line through $u$ and $x$, and let $z$ be the reflection of $v$ in $b$. Note that $z$ lies outside $C_0^x$. Let $D$ be the circle centered at $x$ having radius $|x v|$. Since $x$ lies to the right of $u$, Lemma~\ref{lemma:Yao3Aux} gives us that the left arc $v z$ of circle $D$ is enclosed by circle $C$. Since the area in which $y$ must lie for $x y$ to cross $u v$ is bounded by edge $u v$, the left arc $v z$ of circle $D$, and either the left boundary of $C_0^x$ (if $u \notin C_0^x$) or the line $u x$ (if $u \in C_0^x$), it is enclosed by~$C$. Therefore, there does not exist a point $y \in C_0^x$ such that $x y$ intersects $u v$. 
\end{proof}

\begin{figure}[ht]
    \centering
    \includegraphics{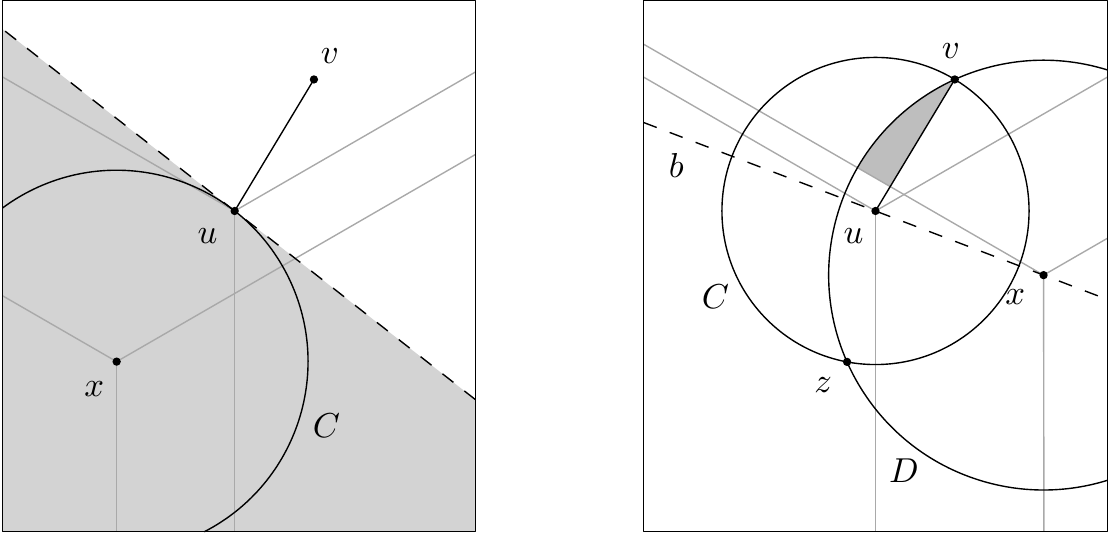}
    \caption{\small Left: The case when $x \in C_2^u$ below $k$ or $x \in C_1^u$ below $l$. 
    Right: The case when $x \in C_1^u$ above $l$ or $x \in C_0^u$ to the right of the line $u v$.}
    \label{fig:Yao3PlanarForestsB}
  \end{figure}

\begin{lemma}
  \label{lemma:Yao3Forest}
Every $i$-path of the \YaoThree-graph is well-defined and has an \isink as one of its endpoints.
\end{lemma}
\begin{proof}
The proof of this lemma is analogous to Lemma~\ref{lemma:Forest} for the $\theta_3$-graph. 
\end{proof}

\begin{lemma}
\label{lemma:Yao3Empty}
If a cone of a point is empty, then no edge in the \YaoThree-graph can cross this cone.
\end{lemma}
\begin{proof}
  We assume without loss of generality that $C_0^u$ does not contain any points. We prove the lemma by contradiction, so assume that there exists an edge $x y$ that crosses $C_0^u$. Since no edge between two points in the same cone can cross another cone, let $x \in C_2^u$ and $y \in C_1^u$. 

Point $y$ cannot lie in $C_0^x$, since either $C_0^x$ does not intersect $C_1^u$ (if $u \notin C_0^x$) or the line segment between $x$ and $y$ does not intersect $C_0^u$ (if $u \in C_0^x$). Hence $y$ must lie in $C_1^x$. 

  If $u \in C_0^x$, $C_1^x$ does not intersect $C_0^u$ and thus the line segment between $x$ and $y$ cannot intersect $C_0^u$ either. Therefore both $u$ and $y$ lie in $C_1^x$. 
  Let $C$ be the circle centered at $x$ with radius $|x u|$. For the edge~$x y$ to exist, $y$ must be closer to $x$ than $u$, which means that $y$ must lie in $C$. Note that $C$ is contained in the half-plane to the left of the tangent to $C$ at $u$.

  If $x$ lies on or above the horizontal line through $u$, the half-plane does not intersect $C_1^u$. If $x$ lies below the horizontal line through $u$, the half-plane does not intersect $C_1^u$ above $u$ and thus $x y$ would not cross $C_0^u$. Since $y$ is enclosed by $C$, $C$ is contained in the half-plane, and there is no point $p$ in the half-plane such that $p \in C_1^u$ and $p x$ crosses $C_0^u$, $x y$ cannot cross $C_0^u$ either. 
\end{proof}

Using Lemmas~\ref{lemma:Yao3Planar}, \ref{lemma:Yao3Forest} and~\ref{lemma:Yao3Empty}, the proof of Theorem~\ref{thm:Theta 3 is connected} translates directly to the \YaoThree-graph yielding the following result.
\begin{theorem}
The \YaoThree-graph is connected.
\end{theorem}
\paragraph{Acknowledgments.} This problem was introduced during the 2012 Fields Workshop on Discrete and Computational Geometry held at Carleton University in Ottawa, Canada. The research of Oswin Aichholzer was partially supported by the ESF EUROCORES programme EuroGIGA - CRP `ComPoSe', Austrian Science Fund (FWF): I648-N18. Work by Sang Won Bae was supported by the Contents Convergence Software Research Center funded by the GRRC Program of Gyeonggi Province, South Korea. The research of Luis Barba, Prosenjit Bose, Andr\'e van Renssen, and Sander Verdonschot was supported in part by NSERC. Matias Korman received support from the Secretary for Universities and Research of the Ministry of Economy and Knowledge of the Government of Catalonia, the European Union, and projects MINECO MTM2012-30951, Gen. Cat. DGR2009SGR1040, ESF EUROCORES programme EuroGIGA -- CRP `ComPoSe': MICINN Project EUI-EURC-2011-4306.

\bibliographystyle{abbrv}
\bibliography{Theta3CGTA}

\end{document}